\documentclass[12pt]{article}
\usepackage{euscript,amsmath,amssymb,latexsym,amsthm}

\newcommand{\be}{\begin{equation}}
\newcommand{\ee}{\end{equation}}
\newcommand{\bea}{\begin{eqnarray}}
\newcommand{\eea}{\end{eqnarray}}

\newcommand{\Rh}{{\mathbb R}}

\DeclareMathOperator{\St}{St}

\theoremstyle{remark}
\newtheorem*{remark}{Remark}

\theoremstyle{plain}
\newtheorem*{theorem}{Theorem}
\newtheorem{proposition}{Proposition}
\newtheorem{lemma}{Lemma}
\newtheorem*{corollary}{Corollary}

\begin{document}
\title{
\vspace{1cm} {\bf Derivation of the Boltzmann equation and entropy production in functional mechanics}
 }
\author{A.\,S.~Trushechkin\bigskip
 \\
{\it  Steklov Mathematical Institute}
\\ {\it Gubkina St. 8, 119991 Moscow, Russia}\medskip\\
{\it  National Research Nuclear University ``MEPhI''}
\\ {\it Kashirskoe Shosse 31, 115409 Moscow, Russia}\bigskip
\\ email:\:\texttt{trushechkin@mi.ras.ru}}

\date {}
\maketitle

\begin{abstract}
A derivation of the Boltzmann equation from the Liouville equation by the use of the Grad limiting procedure in a finite volume  is proposed.
We introduce two scales of space-time: macro- and microscale and use 
the BBGKY hierarchy and the functional formulation of classical mechanics.
According to the functional approach to mechanics, a state of a system of particles is formed from the measurements, which are 
rational numbers. Hence, one can speak about the accuracy of the initial probability 
density function
in the Liouville equation. We assume that the initial data for the microscopic density
functions are assigned by the macroscopic one (so, one can say about a kind of hierarchy and subordination of the microscale to the macroscale) and derive
the Boltzmann equation, which leads to the entropy production.

\end{abstract}

\section{Introduction}

The report concerns the problem of derivation of the kinetic Boltzmann equation from the equations of microscopic dynamics (the Liouville equation). The main interest is related with the fact that the Liouville equation is reversible in time, while the  Boltzmann equation does not. On the contrary, for the Boltzmann equation the so-called $H$-theorem is valid, so, it describes the entropy production and time irreversibility.

The fundamental time irreversibility problem (the Loschmidt's paradox) consists in the following: how to reconcile the time-reversible microscopic dynamics and the irreversible macroscopic one. One of the problems in the framework of the irreversibility problem is the derivation of the Boltzmann equation from the Liouville equation.

An elegant derivation of the Boltzmann equation from the Liouville equation has been proposed by Bogolyubov \cite{Bogol}. He uses the BBGKY (Bogolyubov--Born--Green--Kirkwood--Yvon) hierarchy of equations, thermodynamic limit and some additional assumptions (which do not follow from the Liouville equation). This derivation leads to the divergences in high order corrections to the Boltzmann equation \cite{Bogol77}.

Another derivation of the Boltzmann equation was proposed by Lanford \cite{Lanford}. He also uses the BBGKY hierarchy of equations, the Grad (or Boltzmann--Grad) limit, all assumptions are included in the initial conditions. But the derivation can be applied only to small times. On the contrary, the Boltzmann equation is interesting from the viewpoint of the large time asymptotics (the relaxation to the Maxwell distribution).

We propose a new derivation of the Boltzmann equation. The BBGKY hierarchy of equations are supplemented with the ideas of the functional mechanics, recently proposed by I.\,V.~Volovich \cite{VolFuncMech,VolRand}. Also we introduce two scales of space-time: a macro- and a microscale (a kinetic scale).

A kinetic equation for the system of two particles has been obtained in \cite{VolBogol}. But the obtained equation is time-reversible, as the Liouville equation. In the present work we obtain the irreversible Boltzmann equation, which leads to the entropy production.

\section{Liouville and Boltzmann equations}

Let $N$ particles in a region $G\subset\mathbb
R^3$ with the volume $V$ be given. Their state in an arbitrary moment of time $t$ is described by a function $f(x_1,x_2,\ldots,x_N,t)$, where
$x_i=(q_i,p_i)$, $q_i\in G$ (a position), $p_i\in\mathbb R^3$ (a momentum), $i=1,2,\ldots,N$. Let the normalization condition be $\int f\,dx_1\ldots dx_N=V^N$, i.e.,
$\frac1{V^N}f$ is a probability density function of the $N$-particle system. The dynamics of the function $f$ is given by the Liouville equation
\begin{equation}\label{EqLiouv}
\frac{\partial f}{\partial t}=\{H,f\},
\end{equation}
where
\begin{equation}\label{EqH}
H(q,p)=\sum_{i=1}^N\frac{p_i^2}{2m}+
\sum_{\begin{smallmatrix}i,j=1\\i>j\end{smallmatrix}}^N\Phi(\frac{|q_i-q_j|}\mu)+\sum_{i=1}^NU(q_i)
\end{equation}
is a Hamiltonian of the system, $m>0$ is the mass of a single particle, $\{\cdot,\cdot\}$ are the Poisson brackets. The first term in $H$ corresponds to the kinetic energy of the system, $\Phi(r)$ is the interaction potential of the particles, $U(q_i)$ is an external potential.

We assume that the function $\Phi(r)$ is continuously differentiable, bounded from below, $\Phi(r)\to+\infty$ as $r\to0$, $\Phi(r)\to0$ as $r\to\infty$. Also for simplicity we will assume that $\Phi(r)$ monotonically decreases as $r$ increases (which corresponds to a repelling force). Let the function $U$ be continuously differentiable, $U(q)\to+\infty$ as $q\to\partial G$, where $\partial G$ is the boundary of the region $G$ (so, the external potential does not allow the particles to go beyond $G$).

Further, $\mu>0$ is a small dimensionless parameter. It has a sense of the relation of the interaction radius $r_0$ to the mean free path $l$, i.e., $\mu=\frac{r_0}l\ll1$. This parameter in the Hamiltonian means that the particles interact on the scale much smaller than the scale of kinetic events. 

Time reversibility of the Liouville equation is expressed in terms of the following
\begin{proposition}
Let $f(q,p,t)$ be a solution of the Liouville equation (\ref{EqLiouv}). Let the Hamiltonian satisfy the equality $H(q,p)=H(q,-p)$ (for example, the Hamiltonian has the form (2)). Then $f(q,-p,-t)$ is also a solution of the Liouville equation.
\end{proposition}
Note, however, that in spite of the formal reversibility, the solutions of the Liouville equation obey the property of delocalization which can be regarded as irreversibility \cite{VolFuncMech,VolRand}.

Consider now a gas from the point of view of kinetic theory. In this case the gas is described in terms of a single-particle density function $f(q,p,t)$. Let it be normalized on the volume $V$ (i.e., $f(q,p,t)/V$ is a probability density function), $q\in G$, $p\in\Rh^3$. One of the fundamental equations of kinetic theory is the Boltzmann equation:

\begin{equation}\label{EqBoltz}
\frac{\partial f}{\partial t}=-\frac{p}m\frac{\partial f}{\partial q}+\frac{\partial U}{\partial q}\frac{\partial f}{\partial p}+\St f,\end{equation}
\begin{equation*}
\St f=n\int_{\Rh^2\times\Rh^3}\frac{|p-p_1|}m
[f(q,p',t)f(q,p'_1,t)-f(q,p,t)f(q,p_1,t)]d\sigma dp_1,\end{equation*}
where $n>0$ is the mean concentration of the particles (the mean number of particles in a unit of volume, $d\sigma=rdrd\varphi$. Also $p'$ and $p'_1$ are the momenta that will have two particles long after the collision provided that they had the momenta $p$ and $p_1$ long before the collision with the impact parameter of the collision $r$ and the polar angle $\varphi$. So, $(r,\varphi)\in\Rh^2$ are polar coordinates on the plane perpendicular to the relative velocity vector $(p-p_1)/m$. Thus, $p'=p'(p,p_1,r,\varphi)$, $p'_1=p'_1(p,p_1,r,\varphi)$, the dependence is defined by the two-particle Hamiltonian without the external potential 
\begin{equation}\label{EqH02}
H^0_2=\frac{p_1^2}{2m}+\frac{p_2^2}{2m}+\Phi(\frac{|q_1-q_2|}\mu).
\end{equation}
The expression $\St f$ is called the collision integral.

This is an important nonlinear equation which describes the relaxation of the function $f$ to the Maxwell distribution \cite{Villani}. One of the properties which can be easily proved is the so-called Boltzmann $H$-theorem:

\begin{proposition}[$H$-theorem]
Let f(q,p,t) be a solution of the Boltzmann equation (\ref{EqBoltz}) and the quantity
$$\mathcal H(t)=\int_{\Omega_V}f(q,p,t)\ln f(q,p,t)\,dqdp,$$
$\Omega_V=G\times\Rh^3$, be well-defined (i.e., the integral converges). Then $\frac{d\mathcal H}{dt}\geq0$.
\end{proposition}

This proposition states the entropy production (the quantity $S(t)=-\mathcal H(t)$ can be regarded as the entropy of the gas) and, hence, the irreversible character of the gas dynamics.

Thus, we have obtained two contradictory conclusions: if we consider the gas as a whole (in terms of a single-particle distribution and the Boltzmann equation), the dynamics is irreversible in time, while if we consider the gas as a system of a finite number $N$ of particles (in terms of an $N$-particle distribution function and the Liouville equation), the dynamics is reversible.

It seems that the reductionism does not work, the time irreversibility is a property of macrosystems which cannot be reduced to the microscopic level. So, the problem about another type of relation between the different levels of description arises.

Let us return to the Liouville equation. According to Bogolyubov, let us define the $s$-particle distribution functions $f_s$, $s=1,2,\ldots,N-1$:

$$
f_s(x_1,\ldots,x_s,t)=V^{N-s}\int_{\Omega_V^{N-s}}f(x_1,\ldots,x_N,t)dx_{s+1}\ldots dx_N,
$$
where $\Omega_V=G\times\mathbb R^3$ is the phase space of a single particle, $f_N\equiv f$. Usually, the function $f$ in the Boltzmann equation is associated with the single-particle function $f_1$. This is a common point in the derivations of the Boltzmann equation according to Bogolyubov, Lanford and others. We also follow this idea, but supplement it with the idea of subordination of different space-time scales.

\section{Micro- and macroscale}

Both Bogolyubov and Lanford start with the Cauchy problem for the Liouville equation:

\begin{equation}\label{EqLiouvCauchi}
\left\lbrace\begin{aligned}
&\frac{\partial f}{\partial t}=\{H,f\},\\
&f(x_1,\ldots,x_N,0)=f^0(x_1,\ldots,x_N),
\end{aligned}\right.
\end{equation}
But we should keep in mind that the initial distribution function $f^0$ is not given ``objectively'', but constructed based on measurement results. Note that the measurement results are rational numbers (this is a starting point of p-adic mathematical physics \cite{VVZ,DragovichReview}). See \cite{TrVolFuncRat} for the detailed description of the construction of the probability density function starting from the rational-valued measurement results. See also \cite{TrushTMF10} for the discussion of the functional dynamics of a system under often measurements.

The accuracy of the  measurements is essential here. If our measurement instruments allow us to register variations of the physical values (and, hence, the initial distribution function $f^0$) on the scale of particle interaction radius, the solution of the Cauchy problem has a physical meaning.

But in practice, if we consider the kinetic events, our measurement instruments can register variations of the initial distribution function only on the scales  much larger than the interaction radius \cite{Bogol,Landau10}. The scale of order of the interaction radius $r_0$ will be called ``microscale'', the scale of order of the mean free path $l$ -- ``macroscale'' (or ``kinetic scale''). $\mu=\frac{r_0}l\ll1$ is a scaling parameter. Our measurement instruments can register the variations of the physical values only on the kinetic scale.

In this case we do not know the $N$-particle distribution function $f^0(x_1,\ldots,x_N)$, 
since it reflects the information about correlations of the particles on the scale of order of the interaction radius, which cannot be registered. This information is  essential. For example, if the distance between the particles is microscopically large ($|q_i-q_j|\gg r_0$ for all $i\neq j$), there is no correlation between them and the following factorization property is satisfied:
$$f^0(x_1,x_2,\ldots,x_s,0)=\prod_{i=1}^Nf^0_1(x_i,0).$$
Exactly this  property will be registered by the instrument,  since it can register only the macroscopic variations. However, if we extrapolate this property over the whole phase space, we obtain an infinite mean energy
$E=\frac1{V^N}\int_{\Omega_V} Hf\,dx_1\ldots dx_N$
(since $\Phi(r)\to+\infty$ as $r\to0$).

So, the correlations between the particles are essential on the microscale, but the instrument cannot register them. Hence, we do not know the initial $N$-particle distribution function $f^0$ and the Cauchy problem (\ref{EqLiouvCauchi}) has not a direct physical meaning.

\section{The main theorem}
Let us formulate another problem for the Liouville equation which will be in accordance with the physical meaning. The instrument allows us to establish the initial single-particle distribution function $f_1^0(x_1)$. One can say that $f_1^0$ is a ``macroscopic'' probability distribution, since it varies on the scale much larger than the particle interaction radius. Let $f_1^0$ be continuous with its partial derivatives over each argument.

Let us define the following problem for the Liouville equation:

\begin{equation}\label{EqLiouvHierarchy}
\left\lbrace\begin{aligned}
&\frac{\partial f_\mu}{\partial t}=\{H,f_\mu\},\\
&S^{(2)}_{-\Delta t}[f_{2\mu}(x_1,x_2,t-\Delta t)-f_{1\mu}(x_1,t-\Delta t)f_{1\mu}(x_2,t-\Delta t)]\to0\\
&\quad\text{(as $\mu\to0$, $N\to\infty$, $N\mu^2=const$, and $\Delta t\to0$, $\frac{\Delta t}\mu=\Delta\tau\to\infty$)},\\
&f_{1\mu}(x_1,0)\to f_1^0(x_1)\quad \text{(as $\mu\to0$)},\\
&\int_{\Omega_V^N} Hf_\mu\,dx_1\ldots dx_N<\infty,\\
&f_\mu(x_1,\ldots,x_N,t)=f_\mu(x_{i(1)},\ldots,x_{i(N)},t),\quad (i(1),\dots,i(N))=P(1,\ldots,N),
\end{aligned}\right.
\end{equation}
Here, $S^{(2)}_t$ is the two-particle Hamiltonian flow, i.e., $S^{(2)}_t\varphi(x_1,x_2)=\varphi(x_{1t},x_{2t})$, where $(x_{1t},x_{2t})$ is the phase point $(x_1,x_2)$ moved along the flow defined by the Hamiltonian $H^0_2$ (\ref{EqH02}) on $t$, $P(1,\ldots,N)$ is a permutation of the numbers $1,\ldots,N$. The second condition in (\ref{EqLiouvHierarchy}) is understood as a weak limit \cite{Kozlov} over the variable $p_2$, i.e.,
$$\lim\int_{\mathbb R^3}S^{(2)}_{-\Delta t}[f_{2\mu}(x_1,x_2,t-\Delta t)-f_{1\mu}(x_1,t-\Delta t)f_{1\mu}(x_2,t-\Delta t)]\varphi(p_2)dp_2=0$$
for any function $\varphi$ such that the integral exists. $f_\mu$ depends on $\mu$, since $\mu$ is a parameter in the Hamiltonian.

The limit $\mu\to0$, $N\to\infty$, $N\mu^2=const$ is called the Grad or Boltzmann--Grad limit \cite{Grad} and was used by Lanford.

\begin{theorem}
Let the function $f_\mu(x_1,\ldots,x_N,t)$ satisfy the problem (\ref{EqLiouvHierarchy}) and $f_{1\mu}(x_1,t)$ tend to some function $f_1(x_1,t)$ in the Grad limit in every point $x_1\in\Omega_V$. Let $\Phi(r)$ be monotonically decreasing function
and $\lim\limits_{r\to\infty}r^\gamma\Phi(r)=C\neq0,$ $\gamma>2$. Then the function $f_1(x_1,t)$
satisfies the Boltzmann equation:
\begin{equation*}
\frac{\partial f_1}{\partial t}=-\frac{p}m\frac{\partial f_1}{\partial q}+\frac{\partial U}{\partial q}\frac{\partial f_1}{\partial p}+\St f_1,\end{equation*}
\begin{equation*}
\St f=n\mu^2\int_{\Rh^2\times\Rh^3}\frac{|p-p_1|}m
[f_1(q,p',t)f_1(q,p'_1,t)-f_1(q,p,t)f_1(q,p_1,t)]d\sigma dp_1,\end{equation*}
$f(x_1,0)=f^0_1(x_1)$. Here $n=\frac NV$.
\end{theorem}

Hence, the $H$-theorem and the entropy production (see Proposition 2) are also valid.

Let us discuss the formulation of problem (\ref{EqLiouvHierarchy}). The third condition is the initial data for the single-particle function $f_1$. The fourth condition means the finiteness of the energy, the fifth condition means the symmetry of the density function with respect to permutations (i.e., the particles are indistinguishable).

The most interesting and crucial is the second condition. The knowledge of the initial single-particle function $f^0_1(x_1)$ is not sufficient to get a unique solution for the single-particle function $f_1(x_1,t)$. Since the condition of the form $f(x_1,\ldots,x_N,0)=f^0(x_1,\ldots,x_N)$ has not a direct physical meaning, we must have some additional condition.

As we said above, there are two scales of space-time in this consideration: the microscale, related to the interaction of the particles, and the macroscale, related to the kinetic phenomena. The two-particle function $f_2$ relates to the microscopic scale, because it incorporates the information about the pairwise correlations of the particles on the distances of order of $r_0$. The single-particle function $f_1$ relates to the macroscopic scale, since it does not incorporate the information about the particles' correlations and the kinetic theory is expressed in terms of this function. The second condition in (\ref{EqLiouvHierarchy}) means that the initial value for the microscopic function $f_2$ are assigned by the macroscopic function $f_1$. So, instead of the specification of the initial microscopic function $f^0$, we specify only the initial macroscopic function $f^0_1$ and impose a condition on the microscopic function: in a certain sense it is subordinated to the macroscopic one (in sense that its initial values are assigned from the macroscale).

Let $(q,t)$ be macroscopic space-time variables, $\Delta q\sim l$, $\Delta t\sim l/\overline u$, where $\overline u$ is the mean velocity of the particles, and let $(\xi,\tau)$ be microscopic space-time variables, $\Delta \xi\sim r_0$, $\Delta \tau\sim r_0/\overline u$. These variable are related to each other by the scale transformation
\begin{equation}\label{EqRescal}
\xi=\frac q\mu,\quad\tau=\frac t\mu,\quad\mu=\frac{r_0}l\to0.
\end{equation}
With such a transformation a macroscopically infinitesimal region may be infinitely large from the microscopic point of view. Exactly this situation we can see in the second condition in (\ref{EqLiouvHierarchy}): $\Delta t\to0$, but $\Delta\tau=\frac{\Delta t}\mu\to\infty$. Of course, these two scales meet each other in the collision integral of the Boltzmann equation: a collision is considered as a point and momentary act on the macroscale, but it takes place on the infinite space during the infinite time on the microscale.

The great disparity between microscopic and macroscopic scales as one of the origins of the irreversible macroscopic behaviour was pointed out in \cite{KozTres,Lebowitz}. The used rescaling of the space-time (\ref{EqRescal}) is typical for the derivation of the Boltzmann equation from the Liouville equation \cite{Spohn,DobrSinaiSukhov}. In the case of lattice dynamics this rescaling was used in \cite{DudSpohn,Dud} for the derivation of kinetic and hydrodynamic-type equations. Our proposition is to introduce the subordination of different space-time structures expressed in the form of the second condition in (\ref{EqLiouvHierarchy}), which gives a new way of derivation of the Boltzmann equation.

One can say about a kind of hierarchy and subordination of the scales: the initial values for the processes on the microscale (interactions of the particles) are assigned from the processes on the macroscale (kinetic phenomena). Note that the idea of the hierarchy of times (namely, the microscopic, kinetic and hydrodynamic relaxation times) in a slightly different sense was first proposed by Bogolyubov \cite{Bogol}.

\begin{remark}
We can note that in our limiting process the overall mass of the gas $M=mN$ tends to infinity, since $m$ is constant and $N$ tends to infinity. However, we can also rescale the mass of a single particle in the way $m_\mu=\mu^2m$ and rescale the interaction potential in the way $\Phi_\mu(r)=\mu^2\Phi(r)$ as it was proposed by Grad \cite{Grad} (we substitute $m$ by $m_\mu$ and $\Phi(\frac{|q_i-q_j|}\mu)$ by $\Phi_\mu(\frac{|q_i-q_j|}\mu)$ in (\ref{EqH})). The Hamiltonian system with these additional rescalings can be easily reduced to the considered one. So, the gas is considered as an infinite number of negligible mass particles in the limit. This exactly corresponds to the macroscopic intuition.
\end{remark}

Finally, we would like to note that the functional approach can be useful in the method of molecular dynamics simulations. At present the molecular dynamics follows the Newtonian approach and simulates the movement of material points. On this way it is hard to obtain the properties of complex and, moreover, biological systems, which constitute the aim of the molecular dynamics. Furthermore, the problem of uncontrolled cumulative errors in numerical integration is known. The simulation of the motion in terms of the probability density function seems to be more appropriate for obtaining the properties of the complex and biological systems and more justified from the computational point of view. Also the present report suggests that the initial conditions for the probability density function should be chosen not in an arbitrary way, but rather in the form suggested by the system as a whole. In our case such conditions are given by (\ref{EqLiouvHierarchy}). Let us note that these conditions are statistical in essence, they are formulated for the probability distribution and cannot be reduced to the conditions on the initial positions and momenta of the individual particles.

\section{Proof of the theorem}

The proof of the theorem is divided into four lemmas. Expressions like $A\approx B$ means $A-B\to0$ (in the Grad limit $N\to\infty$, $\mu\to0$, $N\mu^2=const$). Also for simplicity we will skip the lower index $\mu$ of the functions $f_{1\mu}$, $f_{2\mu}$, etc. So,  everywhere in this section $f_s$ means $f_{s\mu}$, $s=1,2,\ldots$.

In fact, the presented theorem and its proof is a rigorous mathematical formulation of the variant of the Bogolyubov's derivation presented in \cite{Landau10}. The distinguishing of two scales of space-time and the Grad limit allow to do this.

\begin{lemma}\label{LemfsDyn}
The functions $f_s(x_1,t)$, $s=1,2,\ldots,N-1$, satisfy the equations
\begin{equation}\label{EqfsDyn}
\frac{\partial f_s}{\partial t}=\{H_s,f_s\}+\frac{N-s}V\int_{\Omega_V}
\{\sum_{i=1}^s\Phi(\frac{|q_i-q_{s+1}|}\mu),f_{s+1}\}\,dx_{s+1}.
\end{equation}
Here
$$
H_s(q,p)=\sum_{i=1}^s\frac{p_i^2}{2m}+
\sum_{\begin{smallmatrix}i,j=1\\i>j\end{smallmatrix}}^s\Phi(\frac{|q_i-q_j|}\mu)+
\sum_{i=1}^sU(q_i)
$$
is the $s$-particle Hamiltonian.

In particular, the function $f_1$ satisfies the equation
$$
\frac{\partial f_1}{\partial t}=-\frac{p_1}m\frac{\partial f_1}{\partial q_1}+
\frac{\partial U(q_1)}{\partial q_1}\frac{\partial f_1}{\partial p_1}+
\frac{N-1}V\int_{\Omega_V}\frac{\partial\Phi(\frac{|q_2-q_1|}\mu)}{\partial q_1}\frac{\partial f_2}{\partial p_1}(x_1,x_2,t)\,dx_2.
$$
\end{lemma}
The proof is straightforward: integration of the Liouville equation over the variables $x_{s+1},\ldots,x_N$. This well-known set of equations is called the BBGKY hierarchy.

\begin{lemma}In the sense of weak limit the following limiting equality is satisfied:
$$f_2(x_1,x_2,t)\approx f_1(X_1,t)f_1(X_2,t),$$
where $X_i=(Q_i,P_i)=(q_{i0}+\frac{p_{i0}\Delta t}m,p_{i0})$, $i=1,2$, $(x_{10},x_{20})=S^{(2)}_{-\Delta t}(x_1,x_2)$, $\Delta t\to0$, $\frac{\Delta t}\mu\to\infty$. Here $x_{i0}=(q_{i0},p_{i0})$, $i=1,2$.
\end{lemma}

\begin{proof}
Firstly, we want to proof the equality
\begin{equation}\label{Eqf2Deltat}
f_2(x_1,x_2,t)\approx f_1(x_{10},t-\Delta t)f_1(x_{20},t-\Delta t),
\end{equation}
Indeed, according to lemma \ref{LemfsDyn}, the function $f_2$ satisfies the equation
$$\frac{\partial f_2}{\partial t}=\{H^0_2,f_2\}+\sum_{i=1}^2\{U(q_i),f_2\}+\frac{N-2}V\int_{\Omega_V}
\{\sum_{i=1}^2\Phi(\frac{|q_i-q_3|}\mu),f_3\}\,dx_3.$$
Hence,
\begin{multline*}
f_2(x_1,x_2,t)=S^{(2)}_{-\Delta t}f_2(x_1,x_2,t-\Delta t)+\\+
\int_{t-\Delta t}^t
\Bigl[
\sum_{i=1}^2\{U(q_i),f_2(x_1,x_2,\tau)\}+\frac{N-2}V\int_{\Omega_V}
\{\sum_{i=1}^2\Phi(\frac{|q_i-q_3|}\mu),f_3(x_1,x_2,x_3,\tau)\}\,dx_3\Bigr]d\tau.
\end{multline*}
Since $\Delta t\to0$, we can neglect the integral term:
$$f_2(x_1,x_2,t)\approx S^{(2)}_{-\Delta t}f_2(x_1,x_2,t-\Delta t)\equiv f_2(x_{10},x_{20},t-\Delta t).$$
By the second condition in (\ref{EqLiouvHierarchy}),
$$f_2(x_{10},x_{20},t-\Delta t)\approx f_1(x_{10},t-\Delta t)f_1(x_{20},t-\Delta t).$$
Equality (\ref{Eqf2Deltat}) has been proved.

Now we apply these arguments once again:
\begin{multline*}
f_1(q_{i0},p_{i0},t-\Delta t)=f_1(q_{i0}+\frac{p_{i0}\Delta t}m,p_{i0},t)-\\-
\int_{t-\Delta t}^t
\Bigl[\{U(q_1),f_1(x_{i0},\tau)\}+\frac{N-2}V\int_{\Omega_V}
\{\Phi(\frac{|q_{i0}-q_3|}\mu),f_2(x_{i0},x_3,\tau)\}\,dx_3\Bigr]d\tau.
\end{multline*}
Since $\Delta t\to0$, we can neglect the integral term:
$$f_1(q_{i0},p_{i0},t-\Delta t)\approx f_1(q_{i0}+\frac{p_{i0}\Delta t}m,p_{i0},t),$$
so,
$$f_2(x_1,x_2,t)\approx f_2(x_{10},x_{20},t-\Delta t)\approx f_1(X_1,t)f_1(X_2,t).$$
\end{proof}

\begin{corollary}
$$
\int_{\Omega_V}\frac{\partial\Phi(\frac{|q_2-q_1|}\mu)}{\partial q_1}\frac{\partial f_2}{\partial p_1}(x_1,x_2,t)dx_2\approx
\int_{\Omega_V}\frac{\partial\Phi(\frac{|q_2-q_1|}\mu)}{\partial q_1}\frac{\partial}{\partial p_1}[f_1(X_1,t)f_1(X_2,t)]dx_2,
$$
and
\begin{equation}\label{EqBogol}
\frac{\partial f_1}{\partial t}\approx-\frac{p_1}m\frac{\partial f_1}{\partial q_1}+
\frac{\partial U(q_1)}{\partial q_1}\frac{\partial f_1}{\partial p_1}+
n\int_{\Omega_V}\frac{\partial\Phi(\frac{|q_2-q_1|}\mu)}{\partial q_1}\frac{\partial}{\partial p_1}[f_1(X_1,t)f_1(X_2,t)]\,dx_2
\end{equation}
($n=\frac NV$, we have replaced the factor $N-1$ by $N$, because the integral has an order of $\mu^2$ and tends to zero).
\end{corollary}

Equation (\ref{EqBogol}) is the so-called Bogolyubov equation. Bogolyubov claims that this equation is more precise than the Boltzmann one and starts with this equation (not from the Boltzmann equation) in order to derive the hydrodynamic equations with viscosity and heat conduction.

\begin{lemma}
\begin{equation*}
\frac{\partial f_1}{\partial t}\approx-\frac{p_1}m\frac{\partial f_1}{\partial q_1}+
\frac{\partial U(q_1)}{\partial q_1}\frac{\partial f_1}{\partial p_1}+
n\int_{\Omega}\frac{\partial\Phi(\frac{|q_2-q_1|}\mu)}{\partial q_1}\frac{\partial}{\partial p_1}[f_1(q_1,p_{10},t)f_1(q_1,p_{20},t)]\,dx_2,
\end{equation*}
where $\Omega=\mathbb R^3\times\mathbb R^3$.
\end{lemma}
\begin{proof}
Since $\Delta t\to0$, we have $Q_i\approx q_i$, $i=1,2$. The function $f_1$ tends to a continuous function as $\mu\to0$, so, $f_1(Q_i,p_{i0},t)\approx f_1(q_i,p_{i0},t)$.

Then, replace the integration domain $\Omega_V$ by $\Omega_r(q_1)=B_r(q_1)\times\mathbb R^3\ni(q_2,p_2)$, where $B_r(q_1)\subset\mathbb R^3$ is a ball with the center in $q_1$ and the radius $r$, $r\to0$, $\frac r\mu\to\infty$. We can do such replacement because of the asymptotic properties of the potential $\Phi$:
\begin{multline*}
\int_{\Omega}\frac{\partial\Phi(\frac{|q_2-q_1|}\mu)}{\partial q_1}\frac{\partial}{\partial p_1}[f_1(q_1,p_{10},t)f_1(q_2,p_{20},t)]\,dx_2\approx\\\approx
\int_{\Omega_r(q_1)}\frac{\partial\Phi(\frac{|q_2-q_1|}\mu)}{\partial q_1}\frac{\partial}{\partial p_1}[f_1(q_1,p_{10},t)f_1(q_2,p_{20},t)]\,dx_2.
\end{multline*}
Now $q_2-q_1\to0$ in the integration domain, since $r\to0$. So, we can replace $f_1(q_2,p_{20},t)$ by $f_1(q_1,p_{20},t)$. Finally, due to the properties of $\Phi$ we can once again replace the integration domain $\Omega_r(q_1)$ by $\Omega$.
\end{proof}

So, we have replaced the spatial argument of the functions $f_1$ by the constant $q_1$. Note that the integrand still depends on $q_2$, because $p_{i0}=p_{i0}(q_1,p_1,q_2,p_2)$.  Note also that the mathematical condition that $f_1$ is a continuous function in the limit $\mu\to0$ means from the physical point of view that the single-particle distribution function varies significantly only on the kinetic scale (the variations on the scale of order of the interaction radius are negligible).

\begin{lemma}
$$n\int_{\Omega_V}\frac{\partial\Phi(\frac{|q_2-q_1|}\mu)}{\partial q_1}\frac{\partial}{\partial p_1}[f_1(q_1,p_{10},t)f_1(q_1,p_{20},t)]\,dx_2\approx\St f_1$$
\end{lemma}
\begin{proof}
Since $q_1$ is a constant in the terms $f_1(q_1,p_{i0},t)$, $i=1,2$, let us introduce the function $g(p,t)=f_1(q_1,p,t)$ and rewrite the integral in the form
$$
\int_\Omega\frac{\partial\Phi(\frac{|q_2-q_1|}\mu)}{\partial q_1}\frac{\partial}{\partial p_1}[f_1(q_1,p_{10},t)f_1(q_1,p_{20},t)]\,dx_2=
\int_\Omega\frac{\partial\Phi(\frac{|q_2-q_1|}\mu)}{\partial q_1}\frac{\partial}{\partial p_1}[g(p_{10},t)g(p_{20},t)]\,dx_2.
$$

The two-particle Hamiltonian without the  external potential can be represented in two ways:
$$H^0_2=\frac{p_1^2}{2m}+\frac{p_2^2}{2m}+\Phi(\frac{|q_2-q_1|}\mu)=
\frac{p_{10}^2}{2m}+\frac{p_{20}^2}{2m}+
\Phi(\frac{|q_{20}-q_{10}|}\mu)\approx\frac{p_{10}^2}{2m}+\frac{p_{20}^2}{2m}.$$
The last limiting equality is satisfied, since $\frac{\Delta t}\mu\to\infty$ and, hence, $\frac{|q_{20}-q_{10}|}\mu\to\infty$. Physically this is just the law of conservation of energy: at the moment $t$ two particles obey the kinetic energy and the potential energy of interaction, at the moment $t-\Delta t$ they are far from each other (relative to their interaction radius), so, they  have only the kinetic energy.

Now, from the one side,
$$\{H^0_2,g(p_{10},t)g(p_{20},t)\}\approx
\{\frac{p_{10}^2}{2m}+\frac{p_{20}^2}{2m},g(p_{10},t)g(p_{20},t)\}=0.$$
From the other side, $p_{10}$ and $p_{20}$ are functions of $q_1,p_1,q_2,p_2$, and
\begin{multline*}
\{H^0_2,g(p_{10},t)g(p_{20},t)\}=
\frac{\partial\Phi(\frac{|q_2-q_1|}\mu)}{\partial q_1}
\frac{\partial}{\partial p_1}[g(p_{10},t)g(p_{20},t)]+
\frac{\partial\Phi(\frac{|q_2-q_1|}\mu)}{\partial q_2}
\frac{\partial}{\partial p_2}[g(p_{10},t)g(p_{20},t)]-\\-\frac{p_1}m\frac{\partial}{\partial q_1}[g(p_{10},t)g(p_{20},t)]-
\frac{p_2}m\frac{\partial}{\partial q_2}[g(p_{10},t)g(p_{20},t)].
\end{multline*}
So,
\begin{multline*}
\frac{\partial\Phi(\frac{|q_2-q_1|}\mu)}{\partial q_1}
\frac{\partial}{\partial p_1}[g(p_{10},t)g(p_{20},t)]+
\frac{\partial\Phi(\frac{|q_2-q_1|}\mu)}{\partial q_2}
\frac{\partial}{\partial p_2}[g(p_{10},t)g(p_{20},t)]-\\-\frac{p_1}m\frac{\partial}{\partial q_1}[g(p_{10},t)g(p_{20},t)]-
\frac{p_2}m\frac{\partial}{\partial q_2}[g(p_{10},t)g(p_{20},t)]\approx0.
\end{multline*}

Express the first term from the others and substitute it to the  integral:
$$
\int_{\Omega}\frac{\partial\Phi(\frac{|q_2-q_1|}\mu)}{\partial q_1}\frac{\partial}{\partial p_1}[g(p_{10},t)g(p_{20},t)]dx_2=
\int_{\Omega}\frac{p_2-p_1}m\frac{\partial}{\partial q}[g(p_{10},t)g(p_{20},t)]dqdp_2,
$$
where $q=q_2-q_1$ (the term with the derivative $\frac{\partial}{\partial p_2}$ vanishes after the transformation of the corresponding integral into the integral over the surface in the momentum space, we use the integrability of the function $g(p,t)$ in $p$). Let us introduce the cylindric coordinates $q=(z,r,\varphi)$, where the axis $z$ is directed along the vector $p_2-p_1$. Then $(p_2-p_1)\frac{\partial}{\partial q}=|p_2-p_1|\frac{\partial}{\partial z}$. By integration over $z$ we get
\begin{multline*}
n\int_{\Omega}\frac{p_2-p_1}m\frac{\partial}{\partial q}[g(p_{10},t)g(p_{20},t)]dqdp_2=
n\int_D\left.\frac{|p_2-p_1|}m[g(p_{10},t)g(p_{20},t)]\right|^{z=+\infty}_{z=-\infty}rdrd\varphi dp_2=\\=
n\mu^2\int_D\frac{|p_2-p_1|}m[g(p'_1,t)g(p'_2,t)-
g(p_1,t)g(p_2,t)]d\sigma dp_2=\St f_1,
\end{multline*}
where
$$p_i=\lim_{z\to-\infty}p_{i0}(x_1,x_2),\quad p'_i=\lim_{z\to+\infty}p_{i0}(x_1,x_2),$$
$i=1,2$, are the momenta before and after the collision correspondingly, $d\sigma=\rho d\rho d\varphi$ is a differential cross-section of the collision, $\rho=\frac r\mu$, $D=\mathbb R^2\times\mathbb R^3$.

This concludes the proofs of the lemma and the theorem.
\end{proof}

\section{The problem of the existence of solution}

The main problem of the presented derivation is the following: whether exists a solution of problem (\ref{EqLiouvHierarchy}) such that $f_{1\mu}(x_1,t)$ tends to a continuous function $f_1(x_1,t)$ in the Grad limit for every moment of time $t$ (this is a condition of the theorem).

This condition is analogous to the Boltzmann's molecular chaos (or molecular disorder) hypothesis and to the Bogolyubov's condition of the absence of correlations between two particles before their collision. In fact, we ask whether exists a solution of the Liouville equation such that the most crucial second condition in (\ref{EqLiouvHierarchy}) is satisfied at every moment of time. Or, equivalently, whether exists such family of initial functions $f^0_{\mu}(x_1,\ldots,x_N)$, $\mu>0$, such that $f^0_{1\mu}$ tends to the predefined function $f^0_1$ as $\mu\to0$ and the second condition in (\ref{EqLiouvHierarchy}) is satisfied at every moment of time $t\geq0$. Intuitively it seems true, but the problem is to prove it rigorously. 

A rigorous proof of this statement could be done by the representation of a solution of the BBGKY hierarchy in the form of the Dyson expansion series. This was used by Lanford \cite{Lanford}. The Dyson expansion series supports some intuition: every term of the expansion is related to the number of collisions. Also in some sense the Dyson series is related to the path integral. But unfortunately, the Dyson series converges only on small times. For this reason the Lanford's derivation of the Boltzmann equation also applies only to small times (see the Introduction).

The fact that the Dyson series can be applied only to small times (less than the mean free time) because of secular terms was noticed also by Bogolyubov \cite{Bogol}. In order to overcome this problem he proposed a more sophisticated method based on some additional assumptions and methods of nonlinear dynamics. But there is no idea how these additional assumptions could be proved. Besides this, as we said in the Introduction, his method leads to the divergences in the high order corrections to the Boltzmann equation.

Note that the limit in the second condition in (\ref{EqLiouvHierarchy}) is formulated in the sense of weak limit. This limit is not satisfied in the pointwise sense for an arbitrary moment of time. But there are arguments that this limit can be satisfied in the weak sense, which is sufficient for us \cite{Bogol77,Shelest}.

Note also that we assume the pointwise convergence of the function $f_{1\mu}(x_1,t)$ to the function $f_1(x_1,t)$ as $\mu\to0$. However, the theorem and the proof are still valid if this convergence also is valid only in the week sense.

\section{Conclusions}
In this report we have proposed a derivation of the Boltzmann equation from the Liouville equation based on some ideas of functional mechanics and measurement theory. The main features of this derivation is the distinguishing of two scales of space-time (micro- and macroscopic or kinetic) and the subordination of the processes on the microscale to the processes on the macroscale.

According to the traditional paradigm of the mathematical physics, the dynamics is completely determined if we know the law of motion, i.e., a differential equation, and the initial values for it. However, the initial values themselves are understood as something external to the equations of mathematical physics (``As regards the present state of the world\ldots the laws of nature are entirely silent'' \cite{Wigner}). We propose another picture: the initial values for a given level of nature are assigned from the higher level.

So, instead of the reductionism, which claims the reducibility of all levels of the nature to the most microscopic level, we propose a hierarchical picture of the world: the lower levels of the nature are subordinated to the higher ones.

\section{Acknowledgments}

The author is grateful to Prof. I.\,V.~Volovich for helpful discussions and remarks. This work was partially supported by the Russian Foundation of Basic Research (project 11-01-00828-a), the grant of the President of the Russian Federation (project NSh-7675.2010.1), and the Division of Mathematics of the Russian Academy of Sciences.

\end{document}